\documentclass[twocolumn,secnumarabic,amssymb,longbibliography]{revtex4-1}

\newcommand {\dsty}{\displaystyle}
\newcommand{\0}{{\bm 0}}
\newcommand{\1}{{\bf 1}}

\newcommand{\A}{\bm{A}}

\newcommand{\D}{\bm{D}}

\newcommand{\Ex}{{\sf E}}

\newcommand{\I}{\bm{I}}

\newcommand{\M}{\bm{M}}

\newcommand{\N}{\bm{N}}

\renewcommand{\P}{\bm{P}}

\newcommand{\Realpart}[1]{\mathrm{Re}({#1})}

\newcommand{\Tc}{{\cal T}}

\newcommand{\T}{\bm{T}}

\newcommand{\Uc}{{\cal U}}

\newcommand{\an}[1]{\begin{align}#1\end{align}}
\newcommand{\ab}[1]{\begin{align*}#1\end{align*}}

\newcommand{\betav}{\bms{\beta}}

\newcommand{\Quote}[1]{\begin{quote}#1\end{quote}}

\newcommand{\bms}[1]{{\boldsymbol{#1}}}
\renewcommand{\bm}[1]{{\boldsymbol #1}}
\newcommand{\df}[2]{\displaystyle \frac{\mbox{\rm d} #1}{\mbox{\rm d} #2}}

\newcommand{\diag}[1]{\mbox{ \bf diag}\matrx{#1}}

\newcommand{\dspfrac}[2]{\frac{\displaystyle #1}{\displaystyle #2} }

\newcommand{\enumlist}[1]{\begin{enumerate}#1\end{enumerate}}

\newcommand{\eqdef}{:=}

\newcommand{\ev}{\bm{e}}
\newcommand{\evt}{{\,\ev^\tp}}
\newcommand{\e}{\bm{e}}

\newcommand{\oldversion}[1]{}

\newcommand{\goesto}{\rightarrow}

\newcommand{\hide}[1]{}	%hides text
\newcommand{\hidex}[1]{\ }	%hides text
\newcommand{\temphide}[1]{}	%hides text
\newcommand{\matrx}[1]{{\left[ \stackrel{}{#1}\right]}}

\newcommand{\muv}{\bms{\mu}}

\newcommand{\pr}{\protect}

\newcommand{\suchthat}{\colon}

\newcommand{\tp}{\top}	% better than \sf, \tt, \rm T, or \textsf{\textsc{t}}

\newcommand{\urlsmall}[1]{{\small\url{#1}}}

\newcommand{\vv}{\bm{v}}

\newcommand{\x}{{\bm{x}}}

\newfont{\gilfont}{cmsy10 scaled\magstep0}
\newcommand{\Naturals}{\mathbb{N}} %natural numbers
\newcommand{\Nz}{\Naturals_0} %natural numbers
\newcommand{\Reals}{\mathbb{R}} % reals
\newcommand{\CNS}{Cosmological Natural Selection}

\usepackage{bm}
\usepackage{amscd}
\usepackage{amsmath}	% Needed!
\usepackage{amsthm}		% Needed!
\usepackage{amsfonts}
\usepackage{amssymb}	
\usepackage{graphicx}	
\usepackage{url}

\newtheorem{Theorem}{Theorem}

\newtheorem*{Implication}{Main Implication}
\newtheorem{Remark}{Remark}

\begin{document}
\title {\LARGE Implications of the Reduction Principle for \\ 
Cosmological Natural Selection
}

\author{Lee Altenberg}%
%\email{altenber@hawaii.edu}
\affiliation{Associate Editor of \emph{BioSystems}; Ronin Institute; \sf{altenber@hawaii.edu}}

\begin{abstract}
Smolin (1992) proposed that the fine-tuning problem for parameters of the Standard Model might be accounted for by a Darwinian process of universe reproduction --- Cosmological Natural Selection (CNS) --- in which black holes give rise to offspring universes with slightly altered parameters.  The laws for variation and inheritance of the parameters are also subject to CNS if variation in transmission laws occurs.  This is the strategy introduced by Nei (1967) to understand genetic transmission, through the evolutionary theory of modifier genes, whose methodology is adopted here.  When mechanisms of variation themselves vary, they are subject to Feldman's (1972) evolutionary Reduction Principle that selection favors greater faithfulness of replication.  A theorem of Karlin (1982) allows one to generalize this principle  beyond biological genetics to the unknown inheritance laws that would operate in CNS.  The reduction principle for CNS is illustrated with a general multitype branching process model of universe creation containing competing inheritance laws.  The most faithful inheritance law dominates the ensemble of universes.  The Reduction Principle thus provides a mechanism to account for high fidelity of inheritance between universes.  Moreover, it reveals that natural selection in the presence of variation in inheritance mechanisms has two distinct objects: maximization of both fitness and faithful inheritance. Tradeoffs between fitness and faithfulness open the possibility that evolved fundamental parameters are compromises, and not local optima to maximize universe production, in which case their local non-optimality may point to their involvement in the universe inheritance mechanisms.
\hide{particle physics -- cosmology connection,
cosmological perturbation theory,
cosmology of theories beyond the SM,
initial conditions and eternal universe}
\end{abstract}
\maketitle
\flushbottom

%%%%%%%%%%%%%%%%%%%%
\section{Smolin's \CNS\ Hypothesis}

\citet{Smolin:1992} has proposed that the free parameters of the Standard Model --- and of the laws of nature, more generally --- may be the result of a process in which gravitational singularities cause new universes to form with slightly different laws of nature, and the entire branching ensemble of universes comes to be dominated by laws which produce the greatest number of gravitational singularities, through a process of 
`cosmological natural selection'.  The elements of the hypothesis are (slightly paraphrased from {\citep{Smolin:2006:Status}}):
\enumlist{
\item The world consists of an ensemble of universes.  
\item Black hole singularities bounce and evolve to initial states of expanding universes. 
\item Hence there is a fitness function equal to the average number of black holes produced by a universe initiated in such a bounce transition.
\item At each such creation event there is a small change in universe properties leading to a small random change to the average number of black holes produced.
\item Under very mild assumptions for the fitness function, the ensemble converges after many universe creation steps to one in which almost every member is near a local maximum in terms of universe production.
\item Therefore,  for a randomly chosen member of the ensemble, almost every small change in the parameters of the Standard Model from its present value either leads the production of black holes unchanged or leads to a decrease. Since our universe can be assumed to be randomly chosen we conclude that if the hypotheses made above are true, almost no change in the parameters of the Standard Model from the present values will increase the numbers of black holes produced.
}

I do not address here the plausibility of any of these elements with respect to the physics.  Critiques and responses in this regard are discussed in \cite{Rothman:and:Ellis:1993,Harrison:1995:Natural,Smolin:1995:Cosmology,Susskind:2004:Cosmic,Susskind:and:Smolin:2004,Silk:1997,Vilenkin:2006:Cosmic,Smolin:2006:Status,Smolin:2007:Scientific,Vilenkin:2007:Problems}.  I also do not address the unsettled issue of how probabilities over multiverse spaces should be interpreted in light of the datum of our universe.  Rather, I address the evolutionary dynamics.

In recent years, the theory of evolutionary dynamics has broadened considerably beyond the phenomena originally identified by Darwin ---  adaptation and diversification.  Augmentations of the 20th century ``Modern Synthesis'' of Darwinian natural selection with Mendelian genetics have grown into a body of theory extensive enough to be called an ``extended synthesis'' \citep{Pigliucci:and:Muller:2010}.  Among the earliest developments in this area was theory to account for the evolution of the genetic mechanisms themselves.

Smolin's proposal leaves mostly open the matter of universe `genetics' --- the mechanisms that might produce changes to the Standard Model parameters from universe creation events.  Some discussion is presented in   \citep{Smolin:2007:Scientific}: 
%%%%
\Quote{The hypothesis that the parameters $p$ change by small random amounts
should be ultimately grounded in fundamental physics. We note that this is
compatible with string theory, in the sense that a great many string vacua
likely populate the space of low-energy parameters.  It is plausible that when
a region of the universe is squeezed to the Planck density and heated to the
Planck temperature, phase transitions may occur, leading to jumps from
one string vacua to another.  But so far there have been no detailed studies
of these processes which would have checked the hypothesis that the change
in each generation is small.  One study of a bouncing cosmology in quantum
gravity also lends support to the hypothesis that the parameters change in
each bounce \citep{Gambini:and:Pullin:2003}.
}

One need not wait for these issues to be resolved before investigating how universe inheritance laws would enter into the evolutionary dynamics.  Moreover, one can ask:  What if universe inheritance laws are themselves subject to evolution?  Here, another point made by Smolin \citep{Smolin:2012:Unification} is relevant:
%%%%
\Quote{[T]he evolution of laws implies a breakdown of the distinction between law and state.  Another way to say this is that there is an enlarged notion of state---a meta-state which codes information needed to specify both an effective law and an effective state,
that the effective law acts on.
}
The `meta-state' to account for the evolution of universe inheritance laws would be information coding for the inheritance laws, which allows them to vary between universes and thus be subject to CNS.

\section{Modifier Theory for the Evolution of Inheritance Mechanisms}
Evolutionary theory in biology took just this approach in trying to understand the evolution of biological inheritance mechanisms when \citet{Nei:1967:Modification} posited the first \emph{modifier gene} model.  Mendel's laws of genetic transmission had hitherto been handled as exogenous parameters in the mathematical models of evolution of the Modern Synthesis.  There was no reason to believe, however, that the material basis for Mendel's laws was any different from the material basis for any other character of organisms.  The development of molecular biology confirmed that the mechanisms of inheritance --- recombination, mutation, DNA repair, etc. --- utilized the same objects as organismal traits involved with survival and fitness: proteins, nucleotides, regulatory sequences, gene interaction networks, and self-organizing structures and activities in the cell and organism.

%%%%%%%%%%%%%%%%%%%%%%%%%%%%%%%%%%%%
Nei's modifier gene concept \cite{Nei:1967:Modification} explicitly ``enlarged notion of state'' with ``a meta-state which codes information needed to specify both an effective law and an effective state.''  The meta-state took the form of a gene posited to control the recombination rate between two other genes under natural selection.  

The first rigorous stability analysis of this model by \citet{Feldman:1972} showed that new forms of the modifier gene could grow in number if and only if they reduced the recombination rate when they occurred in populations near a stationary state.  The same Reduction Principle \citep{Feldman:Christiansen:and:Brooks:1980} was observed for modifiers of mutation rates and dispersal rates as well as recombination (see \cite{Karlin:and:McGregor:1972:Modifier,Teague:1977,Hastings:1983,Altenberg:1984,Feldman:and:Liberman:1986,Altenberg:and:Feldman:1987,Liberman:and:Feldman:1989} for a sample of the early studies).  

The Reduction Principle proposes that populations at a balance between transmission and natural selection will evolve to reduce the rate of error in transmission.  The fundamental implication is that evolution has two different properties that it operates to maximize:  organismal fitness, and the faithfulness of transmission.

The mathematical basis of the Reduction Principle was discovered by \citet{Karlin:1982} (although he did not realize it) in a fundamental theorem on the interaction between \emph{growth} and \emph{mixing}:  
in a system of objects that (1) are changed from one state to another by some transformation processes, and (2) grow or decay in number depending on their state, then \emph{greater mixing produces slower growth or faster decay}.  To be precise, Karlin's theorem states:

\begin{Theorem} [\pr{\citet[Theorem 5.2]{Karlin:1982}}]\label{Theorem:5.2}
Let $\M$ be an irreducible stochastic matrix, and $\D$ a diagonal matrix with positive diagonal elements.  If $\D \neq c \, \I$ for any $c \in \Reals$, then the spectral radius $r( [(1-\alpha) \I + \alpha \M) \D)$ is strictly decreasing in $\alpha \in [0, 1]$.
\end{Theorem}

This phenomenon is not restricted to finite state or discrete time models.  Karlin's theorem has now been extended to general resolvent positive operators on Banach spaces \citep{Altenberg:2012:Resolvent-Positive}.  This relationship between mixing and growth is therefore an extremely general mathematical property.

The implication for the theory of cosmological natural selection is that the parameters of the Standard Model may be optimized not simply for universe replication rate, but also for faithfulness of universe replication, and these two objectives may not coincide, but could potentially involve tradeoffs.  Therefore, any finding that parameters of the Standard Model are not optimal for the rate of universe replication, via the production of black holes in the case of Smolin's proposed mechanism, cannot be taken on face value as falsification of the Cosmological Natural Selection hypothesis.  A further examination of how such parameters may be constrained by the mechanism of universe inheritance, or how they may impact the faithfulness of universe replication, must be considered.

\begin{Implication}
Parameters in the Standard Model that are not local optima for the rate of production of black holes may be involved in the inheritance mechanisms of universe creation and subject to contravening selection for the preservation of  universe properties.
\end{Implication}

%%%%%%%%%%%%%%%%%%%%%%%%%%%%%%%%%%%%
\section{Illustration with a Branching Process Model}

The modifier gene methodology may be applied to investigate the evolution of hypothesized universe inheritance mechanisms through addition of another degree of freedom to the universe state to allow for variation in the inheritance laws.

Variation in the transmission laws can be drawn from a very large choice of alternatives.  Variation in the \emph{magnitude} of changes in the free parameters of the Standard Model is a natural choice.  However, to model the consequences of variation in magnitude one has to include a model of the map from parameters to reproductive properties, a speculative and complicating exercise.  The other principle alternative is to vary the probabilities among a fixed set of changes in state.  The simplest variation in probabilities is to scale equally all of the transitions between states.  Such variation is employed here.

The process of universe creation through black hole creation is naturally modeled as a continuous time branching process with an infinite number of types and with aging to account for individual universe development.  The purpose here, however, is merely to illustrate how universe inheritance mechanisms can evolve, and for this it is sufficient to analyze a simplified branching process with discrete time, finite number of types, and no aging.  The multitype Galton-Watson process is defined as follows.  

First, some notation is introduced.  The whole numbers are represented as $\Nz \eqdef \{0, 1, 2, \ldots \}$.  Vectors are represented as $\x, \vv, \e$, etc.\ and matrices as $\N, \P, \D, \I$, etc..  Matrix elements are $N_{ij} \equiv [\N]_{ij}$.  The vector of ones is $\ev$, its transpose $\evt$, and $\ev_i$ is the vector with $1$ at position $i$ and $0$ elsewhere.  Let $\D_\x \eqdef \diag{x_i}$ represent a diagonal matrix with diagonal elements $x_i$.  A column stochastic matrix $\P$ satisfies $\P \geq \0$ and $\evt \P = \evt$.  An irreducible nonnegative matrix $\A$ has for every pair $(i,j)$ some $t\geq 1$ such that $[\A^t]_{ij} > 0$, and if it is aperiodic, $\A^t > \0$ for some $t\geq 1$.  Letting $\lambda_i(\A)$ be the eigenvalues of $\A$, the spectral radius is $r(\A) = \max_i |\lambda_i(\A)|$, and the spectral bound is $s(\A) = \max_i \Realpart{\lambda_i(\A)}$, where $\Realpart{x}$ is the real part of $x$.

Let $\x \in \Nz^n$ be the $n$-vector of the number of individuals of each of $n$ types.  The probability that a single individual of type $j$ produces a vector of offspring numbers $\x$ is $p(\x, j)$.  Starting with an initial population $\x(0)$ at time $t=0$, each individual produces an offspring vector $\x$ according to $p(\x, j)$, and these vectors are summed to create the population at the next time step, $\x(1)$.  The reproduction process then repeats.  This defines the branching process $\{\x(t) \suchthat t = 0, 1, 2, \ldots \}$.

 The expected number of offspring of type $i$ from parent $j$ is
\ab{
N_{ij} = \sum_{\x \in \Nz^n} x_i \, p(\x, j).
}
It is assumed that $N_{ij}$ exists for all $i,j = \{1, \ldots, n\}$.
The expectation for the process is $\Ex[\x(t) | \x(0)] = \N^t \x(0)$.

The following theorem gives the asymptotic behavior of a super-critical branching process.
%%%%%%%%%%%%%%%%%%%%%%%%%%%%%
\begin{Theorem}[\citet{Kesten:and:Stigum:1966:Limit} \pr{from \cite[Theorem 1, p. 192]{Athreya:and:Ney:1972}}\label{Theorem:KS}]\ 

Let $\x(t) \in \Naturals_0^n$, $t \in \Naturals_0$ be vector of the number of individuals $x_i$ of type $i$ in an n-type Galton-Watson branching process.  Assume the following:
\enumlist{
\item The process is nonsingular, i.e. it is not true that for every $j$ there is an $i$ such that $p(\ev_i, j) = 1$.
\item The process is supercritical, i.e. $r(\N) > 1$.
\item $\N$ is irreducible and aperiodic.
}
Let $\vv(\N)$ be the Perron vector of $\N$ (the eigenvector associated with the spectral bound of $\N$) normalized so $\evt\vv(\N) =1$.

Then (almost surely, a.s.)
\ab{
\lim_{t \goesto \infty} \frac{1}{r(\N)^t} \x_1(t) &= c \, \vv(\N),
}
where $c \geq 0$ is a random variable such that $\Pr[c > 0] > 0$ if and only if for each $i, j \in \{1, \ldots, n\}$
\an{\label{eq:L}
L_{ij} \eqdef \sum_{\x \in \Nz^n} x_i \log(x_i) \, p(\x, j) < \infty.
}
\end{Theorem}

%%%%%%%%%%%%%
Let us now frame a multitype branching process for cosmological natural selection.  Let $\Uc$ represent the set of possible universe types, with a finite number $|\Uc|$ of possible types.  Time is censused in epochs $t = 0, 1, \ldots $.  From one epoch to the next, each universe produces offspring universes according to the multitype Galton-Watson process.  The persistence of a universe between time epochs is accounted for by considering it one of its own offspring.  

Some structure is now introduced for the expectation matrix $\N$.  We decompose $N_{ij}$ in terms of the total expected number of offspring and the distribution among that total.  Let 
\ab{
\beta_j = \sum_{i\in \Uc} N_{ij}
= \sum_{\x \in \Nz^n} \evt\x \, p(\x, j)
}
be the expected total number of offspring from type $j$.  This value $\beta_j$ represents the fitness of universe type $j$.

Let $T_{ij} \geq 0$ be the expected fraction type $i$ among the offspring of $j$, so for each $j$, $\sum_{i\in \Uc} T_{ij} = 1$.  Then $\N$ is decomposed into a fitness component and an offspring distribution component: $N_{ij} = T_{ij} \beta_j$.
Define the matrix $\T \eqdef [T_{ij}]_{i,j \in \Uc}$.  Then in matrix notation:
\ab{
\N = \T \D_\betav.
}
The offspring distribution $\T$ is further decomposed.  The expected proportion of offspring that differ from parent $j$ is given as $\mu_j$.  Thus $1-\mu_j$ is the expected fraction of type $j$'s offspring that are also type $j$, i.e. $T_{jj} = 1 -  \mu_j$.  The expected fraction of offspring of type $j$ that are type $i$ among the non-type-$j$ offspring is $P_{ij}$, hence $T_{ij} = \mu_j P_{ij}$ for $i \neq j$. 
This decomposition in matrix form is
\an{\label{eq:T}
\T = (\I -  \D_\muv) +  \P \D_\muv = \I + (\P - \I) \D_\muv.
}
where $\T$ and $\P$ are stochastic matrices.  It should be noted that many different probability measures $p(\x, j)$ can produce the same matrix of expected proportions $\T$.

%%%%%%%%%%%%%%%%%%%%%%

Now let us include variation in the branching processes.  We enlarge notion of state $i \in \Uc$ to a meta-state $(a, i) \in \Uc' \eqdef \Tc \times \Uc$, where $\Tc$ is the set of possible transmission laws for $\Uc$.  A transmission law for the branching process is the probability measure $p_a(\x,j)$, $p_a \suchthat\Naturals_0^{|\Uc|} \times  \Uc \goesto [0,1]$, that defines the branching process.  An important simplifying assumption made here is that the transmission law states $a \in \Tc$ are themselves transmitted faithfully in offspring universes.  A model in which the transmission laws themselves transform in universe creation is necessary to a more complete theory, but is deferred to further study.  By making the transmission law faithfully replicated, we know that its evolution is due solely to its effects on the universe replication behavior. 

Further, variation in transmission law states $a \in \Tc$ is assumed for now to be \emph{neutral}, in that it leaves unchanged the vector of the mean number of offspring, $\betav$.  

Variation in inheritance laws will be encapsulated solely through a parameter $m_a \in [0,1]$ ($a \in \Tc$) that scales equally the expected proportion of offspring universes that differ from their parents' type.  In the population genetics literature this is referred to as \emph{linear variation} \citep{Altenberg:2009:Linear}.  In matrix form, the scaling by $m_a$ enters into the expected proportion matrix $\T_a$, $a \in \Tc$, from \eqref{eq:T} as
\ab{
\T_a =  (\I -  m_a \D_\muv) +  m_a \P \D_\muv = \I + m_a (\P - \I) \D_\muv.
}

Thus the matrix of the expected numbers of offspring for universes within the transmission class $a \in \Tc$ is
\an{\label{eq:Na}
\N_a = [ \I + m_a (\P - \I) \D_\muv] \D_\betav .
}
\begin{Remark}\rm
A model in which there is variation in the \emph{magnitude} of changes between parent and offspring universes would also be manifest as changes in $\N$, but they would not be linear variation.  For example, let changes in a universe parameter $\gamma_j$ follow a Gaussian distribution (a case considered by \citet{Vilenkin:2007:Problems}), parameterized by $m_a$, to give expected distribution of offspring $i$ having parameter $\gamma_i$,
\an{\label{eq:Gaussian}
T_{ij}(m_a)  = \dspfrac{e^{\dsty -{(\gamma_i-\gamma_j)^2}/{m_a^2}}}{ \sum_{i \in \Uc} \ e^{\dsty -{(\gamma_i-\gamma_j)^2}/{m_a^2}}} .
}
In linear variation, $m_a$ scales all $T_{ij}(m_a)$ equally for $i\neq j$, and disappears from the derivative of $T_{ij}(m_a)$, but clearly this is not the case in \eqref{eq:Gaussian}.  The analysis of $r(\T(m_a) \D_\beta)$ under such non-linear variation is very much an open problem.
\end{Remark}

Now let us examine the evolution of the multiverse ensemble
\ab{
\x \eqdef (\x_1, \x_2, \ldots, \x_{\nu}) \in \Naturals_0^{\nu \times |\Uc|}
}
where there are $\nu$ different transmission laws in the ensemble, where each transmission law $a \in \Tc$ is carried by sub-population $\x_a \in \Naturals_0^{|\Uc|}$.  The matrix of expected numbers for the entire ensemble is a direct sum of diagonal blocks $\N_a$:
\ab{
\N \eqdef \N_1 \oplus \N_2 \oplus \cdots \oplus \N_{\nu} .
}
The diagonal block structure is the result of the transmission laws themselves being faithfully transmitted.

With framework defined, we have the following result.  
%%%%%%%%%%%%%%%%
\begin{Theorem}[Domination of the Most Conservative Transmission Law]\label{Theorem:Main}\ 
\enumlist{
\item Consider an ensemble of universes reproducing according to a branching process, where transmission law $a \in \Tc$ determines the transmission law parameter $m_a \in (0, 1]$ which scales the expected proportion of universes $i \in \Uc$ that differ from their parent universe types $j \in \Uc$.  

\item Let the matrix of expected numbers of offspring universes of type $i$ from a universe of type $j$ be
\ab{
\N_a = [ \I + m_a (\P - \I) \D_\muv] \D_\betav , \text{ for each } a \in \Tc,
}
where $\muv$ is the vector of the fractions of offspring that are not identical to their parents, $m_a$ scales $\muv$, $\betav > \0$ is the vector of expected numbers of offspring universes for each type in $\Uc$, and $\P$ is the matrix of expected distributions of offspring type $i$ from parent type $j \neq i$, $P_{ii} =0$.  

\item In addition make the technical assumption that for each $a\in \Tc$ using $p_a(\x,j)$ in \eqref{eq:L} that $L_{ij} < \infty$ for all $i, j \in \Uc$.

\item Assume that the ensemble branching process is supercritical, so the spectral radius is
$r(\N) = \max_{a \in \Tc} r(\N_a) > 1.$

\item Let stochastic matrix $\P$ be irreducible and aperiodic, and let some $\beta_i \neq \beta_j$.  Let the transmission parameters be ordered so that $m_1 < m_2 < \cdots m_{\nu}$.  
}

Then 
\ab{
\lim_{t \goesto \infty} \frac{1}{r(\N)^t} &(\x_1(t), \x_2(t), \ldots, \x_{\nu}(t)) \\
=
&\ (c_1 \, \vv(\N_1), \0, \ldots, \0) \qquad \text{a.s.},
}
where $c_1$ is as in Theorem \ref {Theorem:KS}.

 Therefore, the ensemble of universes becomes dominated by the transmission law with the highest expected preservation of universe types in universe creation events.
\end{Theorem}
%%%%
\begin{proof}
For the branching process for each transmission type $a \in \Tc$, Theorem \ref{Theorem:KS} gives
\ab{
\lim_{t \goesto \infty} \frac{1}{r(\N_a)^t} \x_a(t) &= c_a \, \vv(\N_a) \ \  \text{a.s..}
}
Suppose that ${r(\N_b)} < {r(\N_a)}$.  Then
\ab{
\lim_{t \goesto \infty} \frac{1}{r(\N_a)^t} \x_b(t) 
&=
\lim_{t \goesto \infty} \frac{r(\N_b)^t}{r(\N_a)^t} \frac{1}{r(\N_b)^t} \x_b(t)  \\
&= c_b \,  \vv(\N_b) \lim_{t \goesto \infty}  \frac{r(\N_b)^t}{r(\N_a)^t}= \0 \ \ \text{a.s..}
}
Thus we need to know when ${r(\N_b)} < {r(\N_a)}$.  This occurs when $m_b > m_a$, as shown by Theorem \ref{Theorem:5.2}.  An equivalent form from \cite{Altenberg:2012:Resolvent-Positive} is that
\an{\label{eq:Reduction}
\df{}{m} s(\D + m \A) < s(\A),
}
where $\A$ is irreducible and has nonnegative off-diagonal elements, and $s()$ refers to the spectral bound, which is the spectral radius if the matrix is nonnegative.  For the case here, $\A = (\P - \I) \D_\muv \D_\betav$, so $r(\A) = 0$ since $\evt (\P - \I) \D_\muv \D_\betav = (\evt - \evt) \D_\muv \D_\betav = 0
$, showing $\evt$ to the be left Perron vector of $\A$.  Thus $r([ \I + m_a (\P - \I) \D_\muv] \D_\betav)$ is strictly decreasing in $m_a$.

Therefore, if $m_a < m_b$ then $r(\N_a)  >  r(\N_b)$.  Hence $m_1 < m_2 < \ldots < m_\nu$ gives $r(\N) = r(\N_1) >  r(\N_2) > \ldots > r(\N_\nu)$.  Thus
\ab{
\lim_{t \goesto \infty}  \frac{r(\N_b)^t}{r(\N)^t} = 0 \qquad \text{for } b \in \{2, 3, \ldots, \nu\}.
}
This produces the claimed result that
\ab{
\lim_{t \goesto \infty}  \frac{1}{r(\N)^t} &(\x_1(t), \x_2(t), \ldots, \x_{\nu}(t)) \\
= \ 
&(c_1 \, \vv(\N_1), \0, \ldots, \0) \text{\ a.s.}.\qedhere
}
\end{proof}
\begin{Remark}\rm
It should be noted that no simplifying assumptions are made regarding the relationship between the expected number of progeny universes $\beta_j$ and the universe inheritance distributions $T_{ij}$.  No constraints are made on the $\beta_j$ other than they be positive, and that $\T$ be irreducible.  The result here is completely general in this regard.  
\end{Remark}

In this branching model, we see that the transmission law with the highest preservation of the laws of physics during universe creation becomes the only typical member of the ensemble as time goes on.  

%%%%%%%%%%%%%%%%%%%%%%%%%%%%%%%%%%%%%%%
\subsection{When Transmission Fidelity is in Conflict with Universe Production}

The variation in universe transmission laws is assumed to be neutral as framed Theorem \ref{Theorem:Main}: $m_a$ varies without varying the expected offspring production $\beta_i$.  However, nature need not produce such a clean separation of effects.  If the transmission type $a \in \Tc$ also affects the reproduction rates, we may give type $a$ its own vector of fitnesses $\betav_a$, and then the matrix of expected offspring numbers becomes
\an{\label{eq:Na2}
\N_a = [ \I + m_a (\P - \I) \D_\muv] \D_{\betav_a} .
}
It is clear that if $m_a < m_b$ then there is some region of $\betav_a$ values that maintains $r(\N_a) > r(\N_b)$, and this range includes values for which $\beta_{ai} < \beta_{bi}$ for each universe type $i \in \Uc$. 

Hence, universe type $(a,i)$ can become dominant over another type $(b,i)$ even though $\beta_{ai} < \beta_{bi}$, due to $m_a < m_b$.  

Therefore, to find a universe type $(b,i) \in \Tc \times \Uc$ with fitness greater than our own universe type $(a,i)$ does not necessarily falsify \CNS, but rather points to the involvement of property $a$ in universe transmission and suggests that transmission law $a$ is more conservative than transmission law $b$.

\section{Additional Limits to Optimization}

Another obvious caveat for predictions from \CNS\ is the assumption that all of the parameters of the Standard Model can be varied in any direction.  Since the origin the parameters is mysterious, and the  hypothesis that they vary as a result of universe reproduction is speculative, the possibility that the parameters are constrained to vary only in some subspace can't be ruled out, in which case optimality arguments have to be made with respect to variation within that subspace and not for a single variable.  Varying a single variable might appear to increase universe reproduction, but not be allowed within the space of variation.  The problem of constraints in variation has a long history in evolutionary theory \citep{Popov:2009:Problem}, as well as the general question of whether evolution optimizes organismal phenotypes \citep{Lewens:2009:Seven}.   

Even in the case where universe creation varies the fundamental parameters in all directions, an additional key point is raised by \citet{Vilenkin:2007:Problems}: evolution produces a distribution over the parameter space when it reaches a mutation-selection balance, and the intensity of mutation pressure compared to selection pressure determines the closeness of a typical universe to a fitness optimum.  Hence, falsifiable predictions cannot be made without some knowledge of the mutation process.  The question of how concentrated or spread out the distribution is around optimal types under mutation-selection balance is a principal issue in the subject of quasispecies theory \citep{McCaskill:1984:Localization,Tejero:Marin:and:Montero:2011:Relationship}.

%%%%%%%%%%%
\section{Discussion}
Theorem \ref{Theorem:Main} shows that the ensemble of branching universes comes to be dominated by universe transmission laws which are the most faithful in preserving the fundamental constants among offspring universes.  This is a manifestation of the Reduction Principle from population biology applied here in a new setting.  

The ubiquitous manifestation of this phenomenon is due to the fundamental mathematics of how mixing interacts with heterogeneous growth.  The most general characterization of this interaction is for operators on Banach spaces in \cite{Altenberg:2012:Resolvent-Positive}, that 
\ab{
\df{}{m} s(m A + V) \leq s(A),
}
where $A$ is resolvent positive operator on a Banach space, $V$ is a real-valued operator of multiplication, $s()$ is the spectral bound, and $m > 0$.  Resolvent positive operators $A$ include Schr\"odinger operators, second order elliptic operators, diffusions, integral operators, and in general, any generator of a positive strongly continuous semigroup.  This theorem derives from Kato's \citeyearpar{Kato:1982:Superconvexity} generalization to Banach spaces of Cohen's \citeyearpar{Cohen:1981:Convexity} theorem that for finite matrices $\A$, $s( \A + \D)$ is convex in diagonal matrices $\D$.

This dynamic has two principle implications for the \CNS\ hypothesis:
\enumlist{
\item  It gives a causal mechanism to incorporate universe transmission laws into the framework of cosmological natural selection, and predicts that transmission of the laws of physics between parent and offspring universes should be very conservative. 
\item It demonstrates that in addition to selection for the maximal number of offspring produced, selection is also indirectly minimizing the probability of change between parent and offspring.  Should there be a tradeoff between these two properties, then the observation of non-optimal values of Standard Model parameters with respect to black hole production may point to the involvement of those parameters in the universe transmission law.
}

The model analyzed here, while more general than a toy model, also leaves out of consideration how transmission laws may be involved in their own transmission.  There are also many other higher-order phenomena in evolutionary dynamics that may prove valuable to include in models of cosmological natural selection, including the evolution of mutational robustness, the evolution of evolvability, and the evolution of modularity in the genotype-phenotype map. 

While the \CNS\ hypothesis remains largely speculative, inclusion of higher-order phenomena  of evolutionary dynamics may provide additional guidance in deriving falsifiable predictions from the theory.

\acknowledgments
I thank Megan Loomis Powers for introducing me to A. Garrett Lisi, whose collaboration with Lee Smolin brought to my attention Smolin's \CNS\ hypothesis.  I thank Laura Marie Herrmann for assistance with the literature search.  I thank Martin Rees, George Ellis, Alexander Vilenkin, and Cl\'ement Vidal for their comments, and Cl\'ement Vidal and Brian D. Josephson for noting typos in an earlier draft.  

\end{document}